\documentclass[11pt]{elsarticle}

\usepackage{amsmath,amssymb,amsfonts,latexsym,stmaryrd}
\usepackage{amsthm}
\usepackage{thmtools,thm-restate}
\usepackage{epsfig}
\usepackage{framed}

\declaretheorem{theorem}
\declaretheorem[sibling=theorem]{lemma}

\newcommand{\beqa}{\begin{eqnarray*}}
\newcommand{\eeqa}{\end{eqnarray*}\par\noindent}

\newcommand{\PP}{\mathcal{P}}

\newcommand{\HH}{\mathcal{H}}
\newcommand{\KK}{\mathcal{K}}

\newcommand{\ket}[1]{{|} #1\rangle}

\newcommand{\ie}{\textit{i.e.}~}

\newcommand{\Complex}{\mathbb{C}}

\newcommand{\Nat}{\mathbb{N}}

\newcommand{\AND}{\; \wedge \;}

\newcommand{\Tr}{\mathsf{Tr}}

\newcommand{\UU}{\mathcal{U}}

\renewcommand{\emph}{\textbf}

\newcommand{\Pn}{\mathcal{P}_n}
\newcommand{\Un}{\mathbf{U}}

\begin{document}

\title{Hardy is (almost) everywhere: \\nonlocality without inequalities for almost all entangled multipartite states}
\author[csox]{Samson Abramsky}
\ead{samson.abramsky@cs.ox.ac.uk}
\author[csox]{Carmen M. Constantin}
\ead{i.m.carmen@gmail.com}
\author[tsing]{Shenggang Ying}
\ead{yingshenggang@gmail.com}
\address[csox]{Department of Computer Science, University of Oxford \\Wolfson Building, Parks Road, Oxford OX1 3QD, U.K.}
\address[tsing]{Department of Computer Science, Tsinghua University, Beijing, China}

\date{}

\begin{abstract}
We show that all $n$-qubit entangled states, with the exception of tensor products of single-qubit and bipartite maximally-entangled states, admit Hardy-type proofs of non-locality without inequalities or probabilities. More precisely, we show that for all such states, there are local, one-qubit observables such that the resulting probability tables are logically contextual in the sense of Abramsky and Brandenburger, this being the general form of the Hardy-type property. Moreover, our proof is constructive: given a state, we show how to produce the witnessing local observables. In fact, we give an algorithm to do this. Although the algorithm is reasonably straightforward, its proof of correctness is non-trivial. A further striking feature is that we show that $n+2$ local observables suffice to witness the logical contextuality of any $n$-qubit state: two each for two for the parties, and one each for the remaining $n-2$ parties.
\end{abstract}

\maketitle

\section{Introduction}

Non-locality theory is a central pillar of both quantum information and quantum foundations. The founding result, Bell's theorem, is Bell's fundamental contribution \cite{bell1964einstein}, which was subsequently given its standard form using the CHSH inequalities by Clauser, Horne, Shimony and Holt \cite{clauser1969proposed}. We shall refer to this as the Bell/CHSH argument. 
This original, and still standard, form of the argument made essential use of the probabilistic predictions of quantum mechanics to separate them from those of any local realistic theory.
Subsequently, Greenberger, Horne, Zeilinger and Shimony \cite{greenberger1989going,greenberger1990bell} gave an important strengthening of Bell's theorem, to an \emph{inequality-free} form, using only quantum predictions made with certainty. The ``price'' for this strengthened form of the theorem is that, while Bell's argument applies to bipartite, 2-qubit systems, the GHZ argument needs at least 3 qubits.
A further important step was taken by Hardy \cite{hardy1992quantum,hardy1993nonlocality}, who showed that an inequality-free proof of Bell's theorem could be given for a 2-qubit system. A curious feature of Hardy's construction is that it works for any bipartite entangled state, \emph{except} for the maximally entangled states.
Subsequently, there have been many extensions and generalisations of these results, e.g. \cite{cabello2001bell,wang2012nonlocality,mansfield2012hardy}.

In \cite{abramsky2011unified}, a general mathematical theory of non-locality and contextuality was developed in a sheaf-theoretic framework. In particular, the content of the Bell/CHSH, GHZ and Hardy arguments was abstracted into general properties of ``empirical models'', \ie of the probability tables summarising the empirical predictions for various combinations of measurements.
An important point that was revealed by this analysis is that the Bell/CHSH, Hardy and GHZ arguments correspond to three different \emph{strengths} or \emph{degrees} of contextual behaviour. We have
\[ \mbox{Bell} \; < \; \mbox{Hardy} \; < \; \mbox{GHZ} . \]
In the terminology of \cite{abramsky2011unified}, the GHZ argument hinges on \emph{strong contextuality}, which can be exhibited by certain quantum systems, and in particular in the usual $n$-qubit multipartite Bell scenarios, as soon as we have 3 or more parties. The Hardy argument hinges on a weaker property, which we termed \emph{logical} or \emph{possibilistic} contextuality in \cite{abramsky2011unified}. This property is in turn stronger than the probabilistic contextuality used in the Bell/CHSH argument.
These notions form a proper hierarchy: strong contextuality is strictly stronger than logical contextuality, which is strictly stronger than probabilistic contextuality.

What distinguishes both logical and strong contextuality from probabilistic contextuality is that they rely only on the \emph{support} of the probability distributions.
That is, these forms of contextuality rely only on whether the probabilities of events are $0$ (impossible) or $> 0$ (possible). In this sense, they are indeed inequality-free \cite{greenberger1990bell,hardy1993nonlocality}, and even probability-free \cite{zimba1993bell,cabello2001bell}. The constructions in the Bell/CHSH, Hardy and GHZ arguments show that certain quantum states and local observables (\ie one-qubit measurements which can be made by each of the parties on their part of the state) give rise to probability tables which are, respectively, probabilistically, logically or strongly contextual.

This leads to the following natural question:

\begin{framed}
%\begin{quotation}
\noindent Given a quantum state, find the maximum level of contextuality it can exhibit for some choice of local observables.
%\end{quotation}
\end{framed}

In particular, we can ask:
\begin{framed}
%\begin{quotation}
\noindent For which quantum states can we find local observables which give rise to a logically contextual (inequality-free, probability-free) form of contextuality?
%\end{quotation}
\end{framed}

In this paper, we shall completely answer this question for $n$-qubit pure states, for all $n$.
Our main result is as follows.

\begin{theorem}
%[Main Theorem]
\label{mainth}
Let $\ket{\psi}$ be an $n$-qubit pure state. Then exactly one of the following two cases must hold:
\begin{enumerate}
\item $\ket{\psi}$ can be written (up to permutation of tensor factors) as a product
\begin{equation}
\label{prodform}
 \ket{\psi} \; = \; \ket{\psi_1} \otimes \cdots \otimes \ket{\psi_k} 
 \end{equation}
where each $\ket{\psi_i}$ is either a 1-qubit state, or a 2-qubit maximally entangled state.
\item There are local observables such that the probability table arising from $\ket{\psi}$ and these local observables is logically contextual, that is, it admits a ``Hardy paradox'', \ie an inequality-free, probability-free proof of non-locality.
\end{enumerate}
If (2) holds, then only $n+2$ local observables are needed; two each for two of the parties, and one each for the other $n-2$ parties.
Moreover, there is an algorithm to decide which of the above cases holds, and which in case (2) explicitly computes the witnessing local observables. The complexity of this algorithm is $O(d \log^3 d)$ in the dimension $d = 2^n$.
\end{theorem}

The remainder of the paper is devoted to the detailed statement and proof of this theorem. In particular, while the algorithm is quite simple, the arguments justifying its correctness are non-trivial.

In Section~2, we define logical contextuality in the $n$-qubit setting. In Section~3, we state the main lemmas, and give the top-level proof of the theorem from these lemmas. In Section~4, we prove the lemmas. In Section~5, we present the algorithm. Finally, in Section~6 we discuss some further directions.

\paragraph{Acknowledgements} We thank Chris Heunen for some helpful discussions and suggestions.

\section{Logical Contextuality}

Our setting consists of a general form of the probability models commonly studied in quantum information and quantum foundations. In these models, a number of agents each has the choice of one of several measurement settings. Moreover each measurement has a number of distinct outcomes. 
We consider here dichotomic measurements with possible outcomes labelled by the elements of $\mathbf{2}:=\{+,-\}$.

Thus, recalling the terminology developed in \cite{abramsky2011unified,abramsky2012logical}, a \emph{measurement scenario} is given by a pair $(X,\mathcal{U})$. 
Here $X$ is a set of measurement labels, while $\mathcal{U}$ is a family of subsets of $X$, giving the maximal sets of compatible measurements, or \emph{contexts}. 
%We consider here dichotomic measurements with possible outcomes labeled by the elements of $\mathbf{2}:=\{+,-\}$.

An empirical model, or generalized probability table, on $(X, \mathcal{U})$ is a family $e=\{e_U\}_{U\in \mathcal{U}}$ of probability distributions $e_U$ on the joint outcomes $\mathbf{2}^U$, one for each context $U$. The probability of obtaining any joint outcome $s\in \mathbf{2}^U$ is given by $e_U(s)$. The support of the model at $U$ is the set $S(U)\subseteq \mathbf{2}^U$ of those joint outcomes $s\in \mathbf{2}^U$ for which $e_U(s)>0$. 

We can make a connection to logic by interpreting the $+$ outcome as \textit{true} and the $-$ outcome as \textit{false}. This allows us to think of the set of measurements $X$ as a set of boolean variables. If $U$ is a finite context, any subset of $\mathbf{2}^U$ can be defined by a propositional formula. For example, each joint outcome $s:U\rightarrow \mathbf{2}$ determines a propositional formula 
$$\varphi_s \; = \; \bigwedge_{x\in U,\, s(x)=+}x \; \AND \; \bigwedge_{x\in U,\,s(x)=-} \neg x$$
The only satisfying assignment of $\varphi_s$ in $\mathbf{2}^U$ is $s$.

The propositional formula whose set of satisfying assignments is the support of $U$ is $\varphi_U:=\bigvee_{s\in S(U)} \varphi_s$.

An empirical model is \emph{logically contextual} if there exists some $U \in \UU$ and some $s\in S(U)$ such that the formula 
$$\Phi \; = \; \varphi_{s} \; \AND \; \bigwedge_{V\in \mathcal{U}\backslash U} \varphi_{V}$$ 
is not satisfiable. 
This says that there is a possible joint outcome $s$ which cannot be accounted for by any valuation on all the variables in $X$ which is consistent with the support of the model. This immediately implies that there is no joint distribution on all the observables which marginalizes to yield the empirically observable probabilities \cite{abramsky2011unified}.
As originally shown in the bipartite case by Fine \cite{fine1982hidden}, and in a very general form in \cite{abramsky2011unified}, the non-existence of a joint distribution is equivalent to the usual definition of non-locality as given by Bell \cite{bell1964einstein}.

An $n$-qubit quantum state $|\psi\rangle$, together with a set of 1-qubit local observables for each $i$, $1 \leq i \leq n$, gives  rise to an empirical model as above. An observable is represented by a $2\times 2$ self-adjoint matrix $M$. The two possible outcomes correspond to its eigenvalues, which we label $+$ and $-$. We use labels of the form $(M,i)$, corresponding to a measurement of observable $M$ performed by the $i^{th}$ party, to denote the elements of the set $X$ of all available measurements. 

A set $U$ of compatible measurements is given by a choice of observable at each of the $n$ measurement sites. A given quantum state $\ket{\psi}$ determines a probability distribution on $\mathbf{2}^U$, as follows.
If 
\[ U=\{(M_1,1),(M_2,2)\ldots (M_n,n)\} \]
is a compatible set of measurements, the probability of obtaining the joint outcome $s:U\rightarrow \mathbf{2}$ is given by the Born rule:
 \[ e_U(s) \; = \; |\langle \ket{\mu_1} \otimes \ket{\mu_2} \otimes \cdots \otimes \ket{\mu_n} |\psi\rangle |^2 , \]
 where $\ket{\mu_i}$ is the eigenvector corresponding to the eigenvalue $s(M, i)$ of the observable $(M, i)$ representing the measurement performed by the $i^{th}$ party. The joint outcome $s$ is in the support of the model if and only if this inner product is non-zero.

\section{The structure of the argument}
\label{Str}

In this section, we shall state a number of main lemmas, and show how Theorem~\ref{mainth} is proved from these lemmas.
The proofs of the lemmas will be given in the following section.

\textbf{Notation}. We shall write $\Pn$ for the set of $n$-qubit pure states of the form~(\ref{prodform}).
For succinctness, we shall write the tensor product of a ket $\ket{\psi}$ with a 1-qubit ket $\ket{i}$, $i=0,1$, as $\ket{\psi}\ket{i}$ rather than $\ket{\psi} \otimes \ket{i}$.

We shall prove Theorem~\ref{mainth} by induction on $n$. The case $n=1$ is trivial.
The $n=2$ case is given by the following lemma.

\begin{restatable}[\textbf{The Base Case Lemma}]{lemma}{basecase}
\label{ground}
Every 2-qubit state $\ket{\psi}$ is either in $\PP_2$ (\ie it is either a product state, or a maximally entangled bipartite state), or there are two local observables for each party, which can be computed directly from the Schmidt decomposition of $\ket{\psi}$, and which witness the logical contextuality of $\ket{\psi}$.
\end{restatable}

The following lemmas and  corollary will be used in the induction step.

\begin{lemma}[\textbf{Going Up Lemma I}]
\label{guI}
If an $n$-qubit state $\ket{\psi}$ is logically contextual with some choice of local observables, then for any $n$-qubit state $\ket{\theta}$, and  $\alpha, \beta \in \Complex$ with $\alpha \neq 0$ and $|\alpha|^2 + |\beta|^2 = 1$, the states
\[ \alpha \ket{\psi} \ket{0} + \beta \ket{\theta} \ket{1}, \qquad \beta \ket{\theta} \ket{0} + \alpha \ket{\psi} \ket{1} \]
are also logically contextual with the same choice of observables, augmented with a single additional observable for the $n+1$-th party.
\end{lemma}

\begin{restatable}[\textbf{Going Up Lemma II}]{lemma}{basecaseII}
\label{lnl}
If $|\theta\rangle=\alpha|\psi\rangle+\beta|\phi\rangle$ is a logically contextual $n$-qubit state under some choice of local observables, then for any  non-zero $x,y \in \Complex$ such that $|x\alpha|^2+|y\beta|^2=1$, the state $|\omega\rangle=x\alpha|\psi\rangle|0\rangle+y\beta|\phi\rangle|1\rangle$ is also logically contextual under the same choice of observables, augmented with a single additional observable for the $n+1$-th party.
\end{restatable}

As an easy consequence of these lemmas, we have:
\begin{restatable}{corollary}{corr}
If $\ket{\theta}$ is logically contextual under some choice of local observables, so are $\ket{\theta} \otimes \ket{\eta}$ and $\ket{\eta} \otimes \ket{\theta}$ for any state $\ket{\eta}$, with the same choice of local observables for each party in $\ket{\theta}$, and a single observable for each party in $\ket{\eta}$. 
\end{restatable}

We now consider the induction step where we have an $n+1$-qubit state $\ket{\omega}$, $n>1$. We can write
\begin{equation}
\label{omeq}
\ket{\omega} \; = \; \alpha \ket{\psi} \ket{0} + \beta \ket{\phi} \ket{1} .
\end{equation}
By the Going Up Lemma I, if either $\ket{\psi}$ or $\ket{\phi}$ are logically contextual, so is $\ket{\omega}$, and we are done.

Suppose now that $\ket{\psi}$ and $\ket{\phi}$ are both in $\Pn$. We consider the parameterised family of states
\begin{equation}
\label{taudef}
\tau(a) \; = \; a \ket{\psi} + \sqrt{1 - a^2} \ket{\phi}, \qquad a \in [0, 1] . 
\end{equation}
If for some $a$, $\tau(a)$ is logically contextual, so is $\ket{\omega}$, by the Going Up Lemma II.
For the remaining case, we have the following rather remarkable result.

\begin{restatable}[\textbf{The Small Difference Lemma}]{lemma}{smalldiff}
\label{smalldifflemm}
Let $\ket{\psi}$ and $\ket{\phi}$ be states  in $\Pn$, and suppose that for all $a \in [0, 1]$, $\tau(a)$ is in $\Pn$, where $\tau(a)$ is defined by~(\ref{taudef}).
Then $\ket{\psi}$ and $\ket{\phi}$ differ in at most one qubit.
\end{restatable}

Applying this result to our decomposition~(\ref{omeq}) of $\ket{\omega}$, we have the following possibilities:
\begin{itemize}
\item $\ket{\psi} = \ket{\phi}$, in which case, from~(\ref{omeq}) and the bilinearity of the tensor product:
\[ \ket{\omega} \; = \; \alpha \ket{\psi} \ket{0} + \beta \ket{\psi} \ket{1} \; = \; \ket{\psi} \otimes (\alpha \ket{0} + \beta \ket{1}) . \]
Since by assumption $\ket{\psi}$ is in $\Pn$, $\ket{\omega}$ is in $\PP_{n+1}$.

\item $\ket{\psi} \neq \ket{\phi}$, in which case (up to permutation) we can write
\[ \begin{array}{lcl}
\ket{\psi} & = & \ket{\Psi} \otimes \ket{\eta} \\
\ket{\phi} & = & \ket{\Psi} \otimes \ket{\theta}
\end{array}
\]
where $\ket{\Psi}$ is in $\PP_{n-1}$ and $\ket{\eta}$ and $\ket{\theta}$ are 1-qubit states.
From this and~(\ref{omeq}), using the bilinearity of tensor product again we have
\[ \ket{\omega} = \ket{\Psi} \otimes \ket{\xi} \]
where $\ket{\xi}$ is a 2-qubit state.
We can apply the Base Case Lemma to $\ket{\xi}$ to conclude that $\ket{\xi}$ is either in $\PP_2$, in which case $\ket{\omega}$ is in $\PP_{n+1}$, or $\ket{\xi}$ is logically contextual, in which case $\ket{\omega}$ is logically contextual by the corollary to the Going Up Lemma.
\end{itemize}

\noindent At this point, we have established (1) and (2) of Theorem~\ref{mainth}, but it seems that we require an infinite search to determine if there exists some $a \in [0, 1]$ for which $\tau(a)$ is logically contextual.

However, the following lemma shows that we only need to test a fixed, finite number of values for $a$ to determine this.

\begin{restatable}[\textbf{The 21 Lemma}]{lemma}{lemmtwoone}
\label{21lemm}
With the same notation as in the Small Difference Lemma, suppose that $\tau(a)$ is in $\Pn$ for 21 distinct values of $a$ in $[0, 1]$. Then $\tau(a)$ is in $\Pn$ for all $a \in [0, 1]$.
\end{restatable}

Thus this lemma allows us to determine which case applies on the basis of a finite number of tests.

In the next section, we shall give proofs of these lemmas. We shall then give an explicit algorithm in Section~5 to complete the proof of Theorem~\ref{mainth}.

\section{Proofs of the lemmas}

Firstly, we collect a few useful basic properties.

\subsection{Background lemmas}

We consider  relations $\sim \; = \; \{ \sim_n \}_{n \in \Nat}$, where $\sim_n$ is an equivalence relation on $n$-qubit states.
We say that $\sim$ is  \emph{LC invariant} if for all $n$-qubit states $\ket{\psi}$, $\ket{\phi}$:
\begin{itemize}
\item If $\ket{\psi} \sim_n \ket{\phi}$, then $\ket{\psi} \in \Pn$ iff $\ket{\phi} \in \Pn$.
\item If $\ket{\psi} \sim_n \ket{\phi}$, then $\ket{\psi}$ is logically contextual iff  $\ket{\phi}$ is logically contextual.
\end{itemize}

\begin{lemma}
\label{permlemm}
The relation induced by permutation of tensor factors is  LC invariant.
\end{lemma}

\begin{lemma}
The relation of LU equivalence is  LC invariant. Here LU equivalence refers to the relation induced by the action of local (1-qubit) unitaries.
\end{lemma}

The following result will be used in the proof of the Small Difference lemma.
It refers to the ``partial inner product'' operation described e.g.~in \cite[p.~129]{QPSI}.\footnote{This is actually  the application of a linear map to a vector under Map-State duality \cite{abramsky2008categorical}.}

\begin{lemma}
\label{tensorthlemm}
Let $\ket{\phi}$ be a state in $\HH \otimes \KK$. 
For any states $\ket{\eta}$ in $\HH$ and $\ket{\theta}$ in $\KK$, if for all $\ket{\eta^{\bot}}$ orthogonal to $\ket{\eta}$,
$\langle \eta^{\bot} | \phi \rangle = \mathbf{0}$,
and  for all $\ket{\theta^{\bot}}$ orthogonal to $\ket{\theta}$,
$\langle \theta^{\bot} | \phi \rangle = \mathbf{0}$,
then (up to global phase) $\ket{\phi}  = \ket{\eta} \otimes \ket{\theta}$.
\end{lemma}
\begin{proof}
We extend $\ket{\eta}$ into an orthonormal basis $\ket{\eta_1}, \ldots , \ket{\eta_n}$ with $\ket{\eta} = \ket{\eta_1}$, and similarly extend $\ket{\theta}$ into $\ket{\theta_1}, \ldots , \ket{\theta_m}$ with $\ket{\theta} = \ket{\theta_1}$. 
Then $B = \{ \ket{\eta_i} \otimes \ket{\theta_j} \}_{i,j}$ forms an orthonormal basis of $\HH \otimes \KK$. Note that
\[ \HH \otimes \KK \; = \; (\ket{\eta} \otimes \ket{\theta})^{\bot \bot} \oplus S^{\bot \bot} \]
where $S = \{ \ket{\eta_i} \otimes \ket{\theta_j} \mid (i,j) \neq (1,1) \}$. Hence $S^{\bot \bot} = (\ket{\eta} \otimes \ket{\theta})^{\bot}$.
By our assumption and the defining property of the partial inner product \cite[Equation (6.47)]{QPSI}, $S \subseteq \ket{\phi}^{\bot}$. Hence
\[ \ket{\phi}^{\bot\bot} \; \subseteq \; S^{\bot} = S^{\bot \bot \bot} = (\ket{\eta} \otimes \ket{\theta})^{\bot \bot} . 
\]
Since these are one-dimensional subspaces, this implies that, up to global phase, $\ket{\phi}  = \ket{\eta} \otimes \ket{\theta}$.
\end{proof}

We now turn to detailed proofs of the main lemmas.
For convenience, we shall repeat the statements of the lemmas.

\subsection{The Base Case lemma}

\basecase*

\begin{proof}
This is essentially Hardy's construction in \cite{hardy1993nonlocality}.
Using the  Schmidt decomposition, every two-particle entangled state can be written in the form
$$|\psi\rangle=\alpha|+\rangle_1|+\rangle_2 + \beta|-\rangle_1|-\rangle_2$$
for an appropriate choice of basis states $|\pm\rangle_i$ for each particle $i$, and normalized non-zero real constants $\alpha$ and $\beta$. 

The logical contextuality of $|\psi\rangle$ is witnessed by a set of four dichotomic observables, two for each of the two parties, namely $U_i$ and $D_i$, $i=1,2$. These observables can be defined as $U_i=|u_i\rangle\langle u_i|$ and $D_i=|d_i\rangle\langle d_i|$ where
\begin{align}
 |u_i\rangle&=\frac{1}{\sqrt{|\alpha|+|\beta|}} (\beta^{\frac{1}{2}}|+\rangle_i + \alpha^{\frac{1}{2}}|-\rangle_i)\\
 |d_i\rangle&=\frac{1}{\sqrt{|\alpha|^3+|\beta|^3}} (\beta^{\frac{3}{2}}|+\rangle_i - \alpha^{\frac{3}{2}}|-\rangle_i)
\end{align}
Hardy's paper also explains why it is not possible to run this particular non-locality argument when either $\alpha$ or $\beta$ are equal to zero (product states), or when $|\alpha| = |\beta|$ (maximally entangled states), that is, when the state $|\psi\rangle$ belongs to $\mathcal{P}_2$. 

There is one subtle remaining point. To show that the dichotomy asserted in the lemma is strictly disjoint, we must show that in the maximally entangled case, there is \emph{no} choice of local observables which can give rise to logical contextuality.
This is shown for the case where each party has the same two local observables as Theorem 3.5 in \cite{abramsky2014classification}, and more generally for any finite sets of local observables as Theorem 2.6.5 in \cite{mansfield2013}.
\end{proof}

\subsection{The Going Up lemmas}

The proofs of the two Going Up lemmas are quite similar. We shall prove the second, which is somewhat harder.

\basecaseII*

\begin{proof}
Since $|\theta\rangle$ is logically contextual there must be some context $U'$ and some $s'\in S(U')$ such that the formula $\Psi=\varphi_{s'}\wedge\bigwedge_{U\in\mathcal{U}\backslash U'} \varphi_U$ is not satisfiable. We will show that it is possible to construct a similar unsatisfiable formula in order to prove the logical non-locality of $|\omega\rangle$.

The $n+1$-th party is assigned a single observable $B = B(x,y)$, whose eigenvectors are $|b_+\rangle=\overline{y}|0\rangle+\overline{x}|1\rangle$ and $|b_-\rangle=x|0\rangle-y|1\rangle$. The observable $B$ is given by the self-adjoint matrix
\begin{equation}\label{Bobservable}
B(x,y)= \left( \begin{array}{cc}
             -|x|^2+|y|^2 & 2x\overline{y} \\ 2\overline{x}y & |x|^2-|y|^2)
            \end{array}\right)  
\end{equation}
For any $n$-qubit state $|\mu\rangle$ we have
\begin{align*}
\langle\mu|\langle b_+|\cdot|\omega\rangle &= x\alpha\langle\mu|\psi\rangle\langle b_+|0\rangle+ y\beta\langle\mu|\phi\rangle\langle b_+|1\rangle\\
&=xy\alpha\langle\mu|\psi\rangle + xy\beta\langle\mu|\phi\rangle\\
&=xy(\alpha\langle\mu|\psi\rangle+\beta\langle\mu|\phi\rangle)\\
&=xy\langle\mu|\theta\rangle
\end{align*}
which implies
\begin{equation}\label{theta1}
\langle\mu|\theta\rangle|^2=0 \Leftrightarrow |\langle\mu|\langle b_+|\cdot|\omega\rangle|^2=0
\end{equation}
The augmented set of allowed measurements $\widetilde{X}=X\cup{(B,n+1)}$ is then covered by the family of compatible subsets 
\[ \widetilde{\mathcal{U}}:=\{\widetilde{U}=U\cup{(B,n+1)}~|~U\in \mathcal{U}\} . \]

Let $T(\widetilde{U})$ denote the support of $\widetilde{U}\in\widetilde{\mathcal{U}}$. Equation~(\ref{theta1}) implies that $S(U)=\{s~|~s+\in T(\widetilde{U}\}$ where $s+ :\widetilde{U}\rightarrow\mathbf{2}$ extends $s$ by mapping $(B,n+1)$ to $+$. As a side remark, note that the presence of an analogously defined section $\sigma-$ in $T(\widetilde{U})$ does not necessarily imply the presence of $\sigma$ in $S(U)$.

Now let $t':=s'+$. We have $t'\in T(\widetilde{U'})$ and we can define the analogue of the proposition $\varphi_{s'}$ as
\begin{align*}
\varphi_{t'}&=\bigwedge_{x\in\widetilde{U'},\, t'(x)=+} x \AND \bigwedge_{x\in\widetilde{U'},\, t'(x)=-}\neg x\\
&=\bigwedge_{x\in U',\,s'(x)=+} x \AND \bigwedge_{x\in U',\,s'(x)=-} \neg x \AND z_{n+1}\\
&=\varphi_{s'}\AND z_{n+1}
\end{align*}
where $z_{n+1}$ denotes the boolean variable corresponding to the outcome on the $n+1$-qubit. Recall that $+$ stands for true and $-$ stands for false.

It also holds that 
\begin{align*}
&\varphi_{\widetilde{U}}=\bigvee_{t\in T(\widetilde{U})}\varphi_t=\bigvee_{t\in T(\widetilde{U}),\,t=s+} \varphi_t \vee \bigvee_{t\in T(\widetilde{U}),\, t=s-}\varphi_t\\
=&\left(\bigvee_{s\in S(U)}(\varphi_s\AND z_{n+1})\right) \vee \left(\bigvee_{t\in T(\widetilde{U}),\, t=\sigma-} \varphi_{\sigma}\AND \neg z_{n+1}\right)\\
=&\left(\left(\bigvee_{s\in S(U)}\varphi_s\right)\AND z_{n+1} \right)\vee \left(\underbrace{\left(\bigvee_{t=\sigma-\in T(\widetilde{U})} \varphi_{\sigma} \right)}_{\gamma_U} \AND \neg z_{n+1}\right)\\
=&\left(\varphi_U\AND z_{n+1}\right)\vee \left(\gamma_U \AND \neg z_{n+1}\right)
\end{align*}

We can now define the formula $\Omega$ which specifies the joint outcome $t'=s'+\in \widetilde{U'}$ as well as the joint outcomes within the supports of all compatible sets of measurements $\widetilde{U}\neq \widetilde{U'}$. This formula is just the conjunction of $\Psi$ and $z_{n+1}$. In order to show that $|\omega\rangle$ is logically contextual, it suffices to show that $\Omega$ has no satisfiable assignment. This is indeed the case, as the fact that $\Psi$ is not satisfiable implies that $\Omega$ is also not satisfiable, thus completing our proof. Indeed, we have
\begin{align*}
\Omega &=\varphi_{t'}\AND\bigwedge_{\widetilde{U}\in \mathcal{\widetilde{U}}\backslash\widetilde{U'}} \varphi_{\widetilde{U'}} \\
&=(\varphi_{s'}\AND z_{n+1})\AND \bigwedge_{\widetilde{U}\in \mathcal{\widetilde{U}}\backslash\widetilde{U'}} \left[ (\varphi_{\widetilde{U}}\AND z_{n+1})\vee (\gamma_U \AND \neg z_{n+1})\right] \\
&=\bigwedge_{\widetilde{U}\in \mathcal{\widetilde{U}}\backslash\widetilde{U'}}\left[ (\varphi_{s'}\AND z_{n+1})\AND ((\varphi_{\widetilde{U}}\AND z_{n+1})\vee(\gamma_U\AND z_{n+1}))\right]\\
&=\bigwedge_{\widetilde{U}\in \mathcal{\widetilde{U}}\backslash\widetilde{U'}}\left[ (\varphi_{s'}\AND z_{n+1})\AND(\varphi_{\widetilde{U}}\AND z_{n+1})\right]\\
&=\varphi_{s'}\AND\left(\bigwedge_{\widetilde{U}\in \mathcal{\widetilde{U}}\backslash\widetilde{U'}} \varphi_{\widetilde{U}}\right)\AND z_{n+1} = \Psi\AND z_{n+1} 
\end{align*}
\end{proof}

The proof of the Going Up Lemma I follows by a similar argument, where the observables are augmented by the  $Z$ measurement for the $n+1$-th party.

The Going Up lemmas have the following useful corollary.

\corr*

\begin{proof}
We argue by induction on the number of qubits in $\ket{\eta}$. 
If $\ket{\eta} = \alpha \ket{0} + \beta\ket{1}$, then by bilinearity of the tensor product,
\[ \ket{\theta} \otimes \ket{\eta} = \alpha \ket{\theta}\ket{0} + \beta \ket{\theta} \ket{1} \]
and we can apply the Going Up Lemma I.

For the inductive case, we can write
\[ \ket{\eta} = \alpha \ket{\eta_0} \ket{0} + \beta \ket{\eta_1} \ket{1} \]
and by bilinearity
\[ \ket{\theta} \otimes \ket{\eta} = \alpha \ket{\theta}\ket{\eta_0} \ket{0} + \beta \ket{\theta} \ket{\eta_1} \ket{1} . \]
By induction hypothesis, $\ket{\theta}\ket{\eta_0}$ and $\ket{\theta}\ket{\eta_1}$ are logically contextual, and we can apply the Going Up Lemma I again to conclude.
\end{proof}

\subsection{The Small Difference lemma}

\smalldiff*

\begin{proof}
Firstly, we note that each state $\ket{\psi}$ in $\Pn$ has an \emph{entanglement type}, which can be described by a graph on $n$ vertices with an edge from $i$ to $j$ when the corresponding qubits of $\ket{\psi}$ are maximally entangled.
There are finitely many such graphs, and we can partition $\Pn$ into $P_1, \ldots , P_M$ according to the entanglement type.
All the states in each $P_l$ are LU equivalent.

As before, we can write $\ket{\psi}$, up to permutation of tensor factors, as
\begin{equation}
\label{psieq}
\ket{\psi} \; = \; \ket{\psi_1} \otimes \cdots \otimes \ket{\psi_k} 
\end{equation}
where each $\ket{\psi_i}$ is either a 1-qubit state, or a 2-qubit maximally entangled state.

We recall the definition of the parameterised family of states $\tau(a)$, $a \in [0, 1]$:
\[ \tau(a) \; = \; a \ket{\psi} + g(a) \ket{\phi}  \]
where $g(a) = \sqrt{1 - a^2}$.

Now for each $\delta \in [0, 1]$, we define a set
\[ R_{\delta}=\{  \tau(a) \mid 1-\delta < a \leq 1\} \]
Under our assumption on $\tau(a)$, each set $R_{\delta}$, which is infinite,  is partitioned among the sets $P_1, \ldots , P_M$.
Also, $\delta < \delta'$ implies $R_{\delta} \supset R_{\delta'}$. Hence, by an application of K\"onig's infinity lemma \cite{levy2012basic}, we can conclude that there is  $l$ with $1 \leq l \leq M$, and  an infinite increasing sequence $\{ a_i \}$ with supremum $1$, such that $\tau(a_i)$ is in $P_l$ for all $i$.

Since all the states $\tau(a_i)$ are LU-equivalent, we can express them in terms of a representative state $\ket{\Theta} \in P_l$ as
\begin{equation}\label{theta}
\tau(a_i) \; = \; U^1_{a_i}\otimes U^2_{a_i}\otimes\ldots\otimes U^n_{a_i} |\Theta\rangle
\end{equation}
Since $\Un(2)^n$ is compact, there is a convergent subsequence $\{U^1_{b_i}\otimes \ldots\otimes U^n_{b_i}\}_{b_i}$, whose limit as $i\rightarrow\infty$ is ${W^1}\otimes\ldots \otimes {W^n}$.
The limit of the corresponding subsequence $\{ a_{b_i} \}$ is still $1$.
Hence $\ket{\psi}$ is also a member of $P_l$, as
\begin{equation*}\label{psi}
|\psi\rangle \; = \; \lim_{i\rightarrow\infty} \tau(a_i) \; = \; \lim_{i\rightarrow\infty} \tau(a_{b_i}) \; = \;  {W^1}\otimes\ldots\otimes {W^n}|\Theta\rangle .
\end{equation*}
This, together with Equation (\ref{theta}), implies that we can express each $\tau(a_i)$ as
\[ \tau(a_i) \;  = \; [(U^1_{a_i}{W^1}^{\dagger})\otimes\ldots\otimes (U^n_{a_i}{W^n}^{\dagger})]|\psi\rangle . \]
Equivalently, using Equation \ref{psieq} and the definition of $\tau(a_i)$, we can obtain a family of equations, one for each $a_i$:
\begin{align}\label{psiphi}
a_i |\psi\rangle + g(a_i)|\phi\rangle &=[(U^1_{a_i}{W^1}^{\dagger})\otimes\ldots\otimes (U^n_{a_i}{W^n}^{\dagger})]|\psi\rangle  \nonumber \\
&= |\psi_1^{i}\rangle\otimes\ldots\otimes |\psi_{k}^{i}\rangle
\end{align}
where $|\psi_{j}^{i}\rangle$ is the state obtained after the LU transformation of $|\psi_j\rangle$, $1 \leq j \leq k$.

We shall now make use of the ``partial inner product'' operation described e.g.~in \cite[p.~129]{QPSI}.  We will use this operation to probe the components of~(\ref{psiphi}).

By Lemma~\ref{tensorthlemm}, if for all $j$, and for any state $\ket{\psi_j^{\perp}}$ orthogonal to $|\psi_j\rangle$, the application of $\ket{\psi_j^{\perp}}$ to $|\phi\rangle$ results in a null vector, we must have  $|\phi\rangle=|\psi\rangle$, and the lemma is proved. 

Otherwise, assume there is some $j$, and some $\ket{\psi_j^{\perp}}$, such that $\langle\psi_j^{\perp}|\phi\rangle$ is a non-zero vector.
For ease of notation, we take $j=k$. 

Applying $\langle\psi_k^{\perp}|$ on both sides of Equation (\ref{psiphi}), we obtain
\begin{align}
g(a_i)\langle\psi_k^{\perp}|\phi\rangle \; &= \; |\psi_1^{i}\rangle\otimes\ldots\otimes |\psi_{k-1}^{i}\rangle \underbrace{\langle\psi_k^{\perp}|\psi_k^{i}\rangle}_{\neq 0} \\
%\mbox{and hence} \\
\underbrace{\langle\psi_k^{\perp}|\phi\rangle}_{constant\ vector} \; &= \; \epsilon_i |\psi_1^{i}\rangle\otimes\ldots\otimes |\psi_{k-1}^{i}\rangle \label{final}
\end{align}
The fact that the LHS of~(\ref{final}) is constant implies that for any $i$ and $j$ we must have $\epsilon_i=\epsilon_j$. We write $\epsilon$ for this common value.
We must also have $|\psi_t^{i}\rangle=|\psi_t^{j}\rangle$ for all $t<k$, and we write $\ket{\psi_t'}$ for the common value. In fact we have 
\begin{equation}\label{limit}
|\psi_t\rangle=\lim_{i\rightarrow\infty}|\psi_t^{i}\rangle=|\psi_t'\rangle, \quad 1 \leq t <k 
\end{equation}
This means that we can rewrite Equation (\ref{final}) as 
\[ \langle\psi_k^{\perp}|\phi\rangle \; = \; \epsilon |\psi_1\rangle\otimes\ldots\otimes |\psi_{k-1}\rangle . \]
A similar analysis will apply to any state $\ket{\eta}$ orthogonal to $\ket{\psi_k}$ for which $\langle \eta | \phi \rangle$ is non-zero.
Using Equation (6.48) from \cite{QPSI}, and bilinearity of the tensor product, we obtain
\begin{equation}\label{form}
|\phi\rangle \; = \; x|\phi_2\rangle\otimes|\psi_k\rangle + \epsilon|\psi_1\rangle\otimes\ldots\otimes|\psi_{k-1}\rangle\otimes \ket{\xi}\
\end{equation}
for some $\ket{\phi_2}$ and $\ket{\xi}$.
We can use Equation (\ref{limit}) to rewrite Equation (\ref{psiphi}) as
\[ a_i |\psi\rangle + g(a_i) |\phi\rangle = |\psi_1\rangle\otimes\ldots\otimes |\psi_{k-1}\rangle\otimes|\psi_{k}^{i}\rangle . \]
For any $j<k$, if we apply any state $\ket{\psi_j^{\perp}}$, orthogonal to $|\psi_j\rangle$, to the above equation we obtain $\langle\psi_j^{\perp}|\phi\rangle=0$. 
Together with Equation (\ref{form}), this implies that $|\phi_2\rangle=|\psi_1\rangle\otimes\ldots\otimes|\psi_{k-1}\rangle$. 
So $|\phi\rangle$ and $|\psi\rangle$ can differ by at most one component. 

These  components cannot be  two-qubit maximally entangled states, since by assumption all linear combinations $\tau(a)$ of $|\psi\rangle$ and $|\phi\rangle$, with $a \in [0, 1]$, belong to $\mathcal{P}_n$, and hence, using bilinearity again, the corresponding linear combinations  of these components  would also have to be maximally entangled, yielding a contradiction.

Thus we conclude that $\ket{\psi}$ and $\ket{\phi}$ can differ at most in a  one-qubit component.
\end{proof}

\subsection{The 21 lemma}

We shall need some elementary facts about partial traces (see e.g.~\cite{nielsen2000quantum})):
\begin{itemize}
\item A pure state in $\HH \otimes \KK$ is a product state  if and only if tracing out over $\HH$ results in a pure state.
\item Tracing out over one party of a maximally entangled bipartite state yields a maximally mixed state.
\item A mixed state $\rho$ is pure if and only if $\Tr \rho^2 = 1$.
\end{itemize}

\lemmtwoone*

\begin{proof}

If a state belongs to $\mathcal{P}_n$ then all partial traces over $n-1$ parties result either in a pure state or in the maximally mixed state $\frac{1}{2}I_2$. 

We can express $|\psi\rangle$ and $|\phi\rangle$ as 
\begin{align*}
 |\psi\rangle \; &= \; \sum_{\sigma_i} a^0_{\sigma_i}|\sigma_i\rangle|0\rangle+ \sum_{\sigma_i} a^1_{\sigma_i}|\sigma_i\rangle|1\rangle\\
 |\phi\rangle \; &= \; \sum_{\sigma_i} b^0_{\sigma_i}|\sigma_i\rangle|0\rangle+ \sum_{\sigma_i} b^1_{\sigma_i}|\sigma_i\rangle|1\rangle
\end{align*}
where the $\sigma_i$ index the elements of the computational basis on $n-1$ qubits.

This means we can write the density matrix corresponding to $|\tau(a)\rangle=a|\phi\rangle + b |\psi\rangle$, with $b = g(a)$, as
\begin{align*}
 |\tau(a)\rangle\langle\tau(a)| \; = & \; \sum_{\sigma_i,\sigma_j}(a a^0_{\sigma_i}+b b^0_{\sigma_i})(a\overline{a^0_{\sigma_j}}+b\overline{b^0_{\sigma_j}})|\sigma_i\rangle\langle\sigma_j|\otimes |0\rangle\langle 0| \\
& \; + \sum_{\sigma_i,\sigma_j}(a a^0_{\sigma_i}+b b^0_{\sigma_i})(a\overline{a^1_{\sigma_j}}+b\overline{b^1_{\sigma_j}})|\sigma_i\rangle\langle\sigma_j|\otimes |0\rangle\langle 1|\\  
& \; + \sum_{\sigma_i,\sigma_j}(a a^1_{\sigma_i}+b b^1_{\sigma_i})(a\overline{a^0_{\sigma_j}}+b\overline{b^0_{\sigma_j}})|\sigma_i\rangle\langle\sigma_j|\otimes |1\rangle\langle 0|\\
& \; +  \sum_{\sigma_i,\sigma_j}(a a^1_{\sigma_i}+b b^1_{\sigma_i})(a\overline{a^1_{\sigma_j}}+b\overline{b^1_{\sigma_j}})|\sigma_i\rangle\langle\sigma_j|\otimes |1\rangle\langle 1|
\end{align*}

The partial trace over the first $n-1$ qubits of $|\tau(a)\rangle$ is given by
\begin{align*}
\rho_n \; = \; & \Tr_{n-1} |\tau(a)\rangle\langle\tau(a)| \; = \; \sum_{\sigma_i} \langle\sigma_i|\tau(a)\rangle\langle\tau(a)|\sigma_i\rangle \\
= \; &\sum_{\sigma_i}(a a^0_{\sigma_i}+b b^0_{\sigma_i})(a\overline{a^0_{\sigma_i}}+b\overline{b^0_{\sigma_i}})|0\rangle\langle 0| + \sum_{\sigma_i}(a a^0_{\sigma_i}+b b^0_{\sigma_i})(a\overline{a^1_{\sigma_i}}+b\overline{b^1_{\sigma_i}})|0\rangle\langle 1| \\
& + \sum_{\sigma_i}(a a^1_{\sigma_i}+b b^1_{\sigma_i})(a\overline{a^0_{\sigma_i}}+b\overline{b^0_{\sigma_i}})|1\rangle\langle 0| + \sum_{\sigma_i}(a a^1_{\sigma_i}+b b^1_{\sigma_i})(a\overline{a^1_{\sigma_i}}+b\overline{b^1_{\sigma_i}})|1\rangle\langle 1| \\
= \; &|0\rangle\langle 0|\sum_{\sigma_i}(a^2 a^0_{\sigma_i}\overline{a^0_{\sigma_i}} +(1-a^2) b^0_{\sigma_i}\overline{b^0_{\sigma_i}} + a\sqrt{1-a^2}(a^0_{\sigma_i}\overline{b^0_{\sigma_i}}+b^0_{\sigma_i}\overline{a^0_{\sigma_i}})) \\ 
\; + \; & |0\rangle\langle 1|\sum_{\sigma_i}(a^2 a^0_{\sigma_i}\overline{a^1_{\sigma_i}} +(1-a^2) b^0_{\sigma_i}\overline{b^1_{\sigma_i}} + a\sqrt{1-a^2}(a^0_{\sigma_i}\overline{b^1_{\sigma_i}}+b^0_{\sigma_i}\overline{a^1_{\sigma_i}})) \\
\; +\; & |1\rangle\langle 0|\sum_{\sigma_i}(a^2 a^1_{\sigma_i}\overline{a^0_{\sigma_i}} +(1-a^2) b^1_{\sigma_i}\overline{b^0_{\sigma_i}} + a\sqrt{1-a^2}(a^1_{\sigma_i}\overline{b^0_{\sigma_i}}+b^1_{\sigma_i}\overline{a^0_{\sigma_i}}))\\
\; +\; & |1\rangle\langle 1|\sum_{\sigma_i}(a^2 a^1_{\sigma_i}\overline{a^1_{\sigma_i}} +(1-a^2) b^1_{\sigma_i}\overline{b^1_{\sigma_i}} + a\sqrt{1-a^2}(a^1_{\sigma_i}\overline{b^1_{\sigma_i}}+b^1_{\sigma_i}\overline{a^1_{\sigma_i}})) 
\end{align*}

The partial trace $\rho_n$ is equal to the maximally mixed state if and only if 
\begin{align*}
 1/2 \; &= \; \sum_{\sigma_i}(a^2 a^0_{\sigma_i}\overline{a^0_{\sigma_i}} +(1-a^2) b^0_{\sigma_i}\overline{b^0_{\sigma_i}} + a\sqrt{1-a^2}(a^0_{\sigma_i}\overline{b^0_{\sigma_i}}+b^0_{\sigma_i}\overline{a^0_{\sigma_i}})) \\
0 \; &= \; \sum_{\sigma_i}(a^2 a^0_{\sigma_i}\overline{a^1_{\sigma_i}} +(1-a^2) b^0_{\sigma_i}\overline{b^1_{\sigma_i}} + a\sqrt{1-a^2}(a^0_{\sigma_i}\overline{b^1_{\sigma_i}}+b^0_{\sigma_i}\overline{a^1_{\sigma_i}}))\\
0\; &= \; \sum_{\sigma_i}(a^2 a^1_{\sigma_i}\overline{a^0_{\sigma_i}} +(1-a^2) b^1_{\sigma_i}\overline{b^0_{\sigma_i}} + a\sqrt{1-a^2}(a^1_{\sigma_i}\overline{b^0_{\sigma_i}}+b^1_{\sigma_i}\overline{a^0_{\sigma_i}}))\\
1/2 \; &= \; \sum_{\sigma_i}(a^2 a^1_{\sigma_i}\overline{a^1_{\sigma_i}} +(1-a^2) b^1_{\sigma_i}\overline{b^1_{\sigma_i}} + a\sqrt{1-a^2}(a^1_{\sigma_i}\overline{b^1_{\sigma_i}}+b^1_{\sigma_i}\overline{a^1_{\sigma_i}})) 
\end{align*}

Each of these equations yields a polynomial of degree 4 in $a$. Indeed, each equation has the form
\[ cab + q(a) = d \]
where $c$ and $d$ are constants, $b = \sqrt{1 - a^2}$, and $q(a)$ is a quadratic polynomial in $a$.
We can write this as
\[ cab  = -q(a) + d \]
and square both sides to obtain
\[ c^2a^2(1 - a^2) = (-q(a) + d)^2 \]
which is a quartic polynomial in $a$.
Hence, there can be at most 4 values of $a$  in $[0, 1]$ for which the partial trace $\rho_n$ is equal to a maximally mixed state.

On the other hand, $\rho_n$ is equal to a pure state if and only if $\Tr\rho_n^2=1$. By a similar analysis, this condition turns out to be equivalent to a polynomial equation of degree 16 in $a$. Unless the polynomial is degenerate, \ie the coefficients cancel so that the equation reduces to $1=1$, there can be at most 16 values of $a$ for which the partial trace $\rho_n$ is equal to a pure state.

Therefore, if there are more than $4+16=20$ values of $a$ in $[0, 1]$ for which the linear combination $\tau(a)$ belongs to $\mathcal{P}_n$, we can conclude that one of the polynomial equations above was degenerate, hence $\tau(a) \in\mathcal{P}_n$ for all values of $a$. 
\end{proof}

\paragraph{Remark} If $\ket{\psi}$ and $\ket{\phi}$ differ by one qubit, the partial trace of $\tau(a)$ for any value of $a$ will yield a pure state, and hence we will always have $\Tr\rho_n^2=1$. Thus the polynomial will indeed be degenerate in this case. This shows the necessity for the Small Difference Lemma.

\section{The algorithm}

We now give an explicit, albeit informal description of the algorithm which follows straightforwardly from our results.

We begin with a subroutine which we will use to test if a state is in $\Pn$.

\vspace{.1in}
\begin{tabular}{ll}
\textsc{subroutine} & \textsf{Test}$\Pn$  \\

\textbf{Input} & $n$-qubit quantum state $\ket{\theta}$ \\
\textbf{Output} & Either \\
& \textsf{Yes}, and entanglement type of $\ket{\theta}$, or \\
& \textsf{No}
\end{tabular}

\begin{enumerate}
\item Compute the $n-1$ partial traces $\rho_i$ over $n-1$ qubits of $\ket{\theta}$.
If any $\rho_i$ is not a maximally mixed state, compute $\Tr \rho_i^2$.
If $\Tr \rho_i^2 \neq 1$, return \textsf{No}.

We now have the list $\{ i_1, \ldots , i_k \}$ of indices for which the maximally mixed state was returned.

\item For each $i_p$ in the list, find its ``partner'' $i_q$ by computing the partial traces $\rho_{i_p,i_q}$ over $n-2$ qubits, and then testing if $\Tr \rho^2_{i_p,i_q} = 1$.\\
If we cannot find the partner for some $i_p$, return \textsf{No}.

\item Otherwise, we return \textsf{Yes}.
We also have the complete entanglement type of $\ket{\theta}$, and we have computed all the single-qubit components. $\Box$
\end{enumerate}

\begin{tabular}{ll}
\textsc{algorithm} & \\
\textbf{Input} & An $n$-qubit state $\ket{\omega}$ \\

\textbf{Output} & Either \\
& \textsf{Yes} if $\ket{\omega}$ is logically contextual, \\ 
& together with a list of $n+2$ local observables, or \\
& \textsf{No} if $\ket{\omega}$ is in $\Pn$.
\end{tabular}

\subsection*{Base Cases}

\begin{enumerate}
\item If $n=1$, output \textsf{No}.
\item If $n=2$, apply the Hardy procedure of the Base Case Lemma to the Schmidt decomposition of $\ket{\omega}$.
\end{enumerate}

\subsection*{Recursive Case: $n+1$, $n>1$}

\begin{enumerate}
\item We apply \textsf{Test}$\PP_{n+1}$ to $\ket{\omega}$. If $\ket{\omega}$ is in $\PP_{n+1}$, return \textsf{No}.

\item Otherwise,
we write
\[ \ket{\omega} \; = \; \alpha \ket{\psi} \ket{0} + \beta \ket{\phi} \ket{1} . \]
Explicitly, if $\ket{\omega}$ is represented by a $2^{n+1}$-dimensional complex vector
\[ \sum_{\sigma \in \{ 0, 1 \}^{n+1}} a_{\sigma} \ket{\sigma} \]
in the computational basis, we can define
\[ \alpha = \sqrt{\sum_{\sigma \in \{ 0, 1 \}^n} | a_{\sigma 0} |^2}, \qquad \beta = \sqrt{\sum_{\sigma \in \{ 0, 1 \}^n} | a_{\sigma 1} |^2} \]
\[ \ket{\psi} = \frac{1}{\alpha} \sum_{\sigma \in \{ 0, 1 \}^n} a_{\sigma 0} \ket{\sigma}, \qquad \ket{\phi} = \frac{1}{\beta} \sum_{\sigma \in \{ 0, 1 \}^n} a_{\sigma 1} \ket{\sigma} . \]

\item We apply \textsf{Test}$\Pn$ to $\ket{\psi}$. If $\ket{\psi}$ is not in $\Pn$, we proceed recursively with $\ket{\psi}$, and then extend the observables using the construction of the Going Up Lemma I. 

\item Otherwise, we proceed similarly with $\ket{\phi}$.

\item Otherwise, both $\ket{\psi}$ and $\ket{\phi}$ are in $\Pn$. \\
For $a$ in $(0, 1)$, we define
\[ \tau(a) \; := \; a \ket{\psi} + \sqrt{1 - a^2} \ket{\phi} . \]
For $19$ distinct values in $(0, 1)$, we assign these values to $a$, and apply \textsf{Test}$\Pn$ to $\tau(a)$.

If we find a value of $a$ for which $\tau(a)$ is not in $\Pn$, we use that value to compute the local observable $B(\frac{\alpha}{a},\frac{\beta}{\sqrt{1-a^2}})$ for the $n+1$-th party, as specified in the Going Up Lemma II, and continue the recursion with the $n$-qubit state $\tau(a)$.

\item Otherwise, by the 21 Lemma and the Small Difference Lemma, the only remaining case is where $\ket{\psi}$ and $\ket{\phi}$ differ in one qubit. We have these qubits $\ket{\psi_1}$, $\ket{\phi_1}$ from our previous applications of \textsf{Test}$\Pn$.
In this final case, we can write $\ket{\omega}$ as
\[ \ket{\omega} = \ket{\Psi} \otimes \ket{\xi} \]
where $\ket{\Psi}$ is in $\PP_{n-1}$, and $\ket{\xi}$ is a 2-qubit state. 
Moreover, we have 
\[ \ket{\xi} = \alpha \ket{\psi_1}\ket{0} + \beta \ket{\phi_1}\ket{1} . \]

\item We apply the Base Case procedure to $\ket{\xi}$, which we know cannot be maximally entangled, by Step 1.
We output \textsf{Yes}, together with the two local observables for each party produced by the Hardy construction, and the $n-2$ local observables for $\ket{\Psi}$ produced by the Corollary to the Going Up lemmas. $\Box$
\end{enumerate}

The above algorithm of course involves computation over the real and complex numbers. More precisely, with the usual coding of complex numbers as pairs of reals, we require the field operations and comparison tests  on real numbers. For simplicity, we discuss the complexity of the algorithm in the Blum-Shub-Smale model of computation \cite{blum1989theory}, where we assume that arbitrary real numbers can be stored, and the above operations performed, with unit cost.
Thus the input size of an $n$-qubit state is the dimension $d = 2^n$.

The \textsf{Test}$\Pn$ subroutine performs $n-1$ partial traces over $n-1$ qubits. Each such partial trace involves computing the 4 entries of a matrix, where each entry is a sum over $2^{n-1}$ products. It also computes $O(n^2)$ partial traces over $n-2$ qubits, each of which involves computing 16 entries, each a sum over $2^{n-2}$ products. 
For an $n$-qubit input, at each level of the recursion, we call the subroutine a number of times bounded by a constant, and the recursion terminates in $O(n)$ steps.
Thus we obtain a complexity bound of $O(d \log^3 d)$ operations, in the input size $d$. Of course, in practice the limiting factor is the exponential size of the classical  representation of a quantum state.

The algorithm has been implemented in \textsf{Mathematica}\texttrademark, and has been tested on input states of up to 10 qubits.

\section{Final Remarks}

Our results provide a complete answer to the question we posed in the Introduction, as to which quantum states give rise to logical contextuality, for the case of pure $n$-qubit states. There are natural generalisations to qudits and mixed states, which we leave to future work. 

Given the paucity of structural results on multipartite entanglement, we feel that the perspective offered by the question of the highest degree of contextuality which states can achieve may be fruitful. This is particularly so in the light of recent results showing that contextuality is a key ingredient enabling quantum computation \cite{raussendorf2013contextuality,howard:14}.
The fact that our results apply to \emph{all} multipartite entangled states, with certain bipartite exceptions, is striking.
It shows that the probability- and inequality-free, logical formulation of contextuality and non-locality is not a rare occurrence, but in fact arises for almost all states.
We also note that these qualitative arguments can easily be turned into quantitative ones. It is shown in \cite{abramsky2012logical} how any instance of logical contextuality gives rise to a Bell inequality based on logical consistency conditions, which allows for  quantitative, robust experimental tests.
Our results also add force to the argument, which we have made elsewhere \cite{abramsky2014classification}, that the bipartite case, while by far the most studied in non-locality theory, may actually be anomalous. This certainly proves to be the case for logical contextuality, as our results show.

Moreover, if we consider the notion of \emph{strong contextuality}, also introduced in \cite{abramsky2011unified}, we obtain additional evidence in this direction.
We recall that an empirical model is strongly contextual if the set of formulas corresponding to the supports for the various contexts are not simultaneously satisfiable; that is, there is no assignment whatsoever to all the measurement variables which is consistent with the constraints imposed  by the model.
This is clearly stronger than logical contextuality, which only requires that there is \emph{some} local joint outcome which cannot be extended to the whole set of variables. Strong contextuality is a natural and salient notion with a number of equivalent characterisations. Again, the bipartite case proves anomalous.
In fact, the only strongly contextual bipartite models are the PR-boxes \cite{lal:11,abramsky2011unified}, which are of course not quantum realizable \cite{popescu1994quantum}. By contrast, for all $n>2$, the $n$-partite GHZ states are strongly contextual \cite{abramsky2011unified}.

This leads us on to the \emph{main question} which is the natural next challenge following on from the one we have addressed in this paper:

\begin{framed}
%\begin{quotation}
\noindent For which quantum states can we find local observables which give rise to a strongly contextual empirical model?
%\end{quotation}
\end{framed}

This question remains open, and appears difficult. It has an interesting relation to the question of ``All-versus-Nothing'' arguments \cite{mermin1990simple}, which have recently been studied in the sheaf-theoretic approach \cite{abramsky2015contextuality}, and shown to be related to the cohomological witnesses for contextuality previously introduced in \cite{abramsky:11a}.
All currently known examples of strong contextuality arising in quantum mechanics come from All-versus-Nothing arguments. Determining whether this is true in general is another challenging problem, which may hold the key to the main question.

\section*{References}

\bibliographystyle{elsarticle-num}
\bibliography{bdbib}
\end{document}